\documentclass[11pt,a4paper]{article}
\usepackage[utf8]{inputenc}
\usepackage[english]{babel}
\usepackage[T1]{fontenc}
\usepackage[dvipsnames]{xcolor}
\usepackage[unicode=true, colorlinks = true, allcolors=MidnightBlue]{hyperref}
\usepackage{mathtools}
\usepackage{amsfonts}
\usepackage{amssymb}
\usepackage{amsthm}
\usepackage{mathrsfs}

\usepackage{natbib}
\setcitestyle{citesep={,}}

\usepackage{float}

\usepackage[inline]{enumitem}
\usepackage{graphicx}
\usepackage[margin=1in]{geometry}
\usepackage{marginnote}

\usepackage{nicefrac}
\usepackage{setspace}
\usepackage{thmtools}

\usepackage{cancel}

\usepackage{cleveref}

\usepackage{shellesc}
\usepackage{pgfplots}
\pgfplotsset{compat=1.14}

\newtheorem{theorem}{Theorem}
\newtheorem{defi}{Definition}
\newtheorem{propo}{Proposition}

\theoremstyle{definition}

\newtheorem{ex}{Example}

\theoremstyle{remark}
\newtheorem{remark}{Remark}

\newcommand{\A}{{\mathcal{A}}}
\newcommand{\R}{\ensuremath{\mathbb{R} }}

\newcommand{\E}{\ensuremath{\mathbb{E} }}
\newcommand{\Pro}{\ensuremath{\mathbb{P}}}
\newcommand{\Li}{\ensuremath{L^\infty} }
\newcommand{\az}{\ensuremath{\{\A_Z: Z \in \A_\p\}} }
\newcommand{\Lp}{\ensuremath{L^p} }

\newcommand{\p}{\ensuremath{\rho }}
\newcommand{\X}{\ensuremath{{\Li} }}

\newcommand{\Ap}{{\mathcal{A}_\p}}
\newcommand{\B}{{\mathcal{B}}}
\newcommand{\cl}{\operatorname{cl}}
\newcommand{\pa}{{\p_\A}}

\DeclareMathOperator*{\esssup}{ess\,sup}

\title{On the link between monetary and star-shaped risk measures}
\author{Marlon Moresco \\ email: \href{mailto:marlonmoresco@hotmail.com}{marlonmoresco@hotmail.com} 
 \and Marcelo Brutti Righi  \\ email: \href{marcelo.righi@ufrgs.br}{marcelo.righi@ufrgs.br}  }
\begin{document}
\maketitle

\onehalfspacing
\begin{abstract}
Recently, \cite{castagnoli21} introduce the class of star-shaped risk measures as a generalization of convex and coherent ones, proving that there is a representation as the pointwise minimum of some family composed by convex risk measures. Concomitantly, \cite{jia21} prove a similar representation result for monetary risk measures, which are more general than star-shaped ones. Then, there is a question on how both classes are connected. In this letter, we provide an answer by casting light on the importance  of the acceptability of  $0$, which is linked to the property of normalization. We then show that under mild conditions, a monetary risk measure is only a translation away from star-shapedness.

\end{abstract} 
 
\section{Introduction}

Since the seminal paper of \citep{artzner99}, axiomatic theory for risk measures has been receiving increasing attention. This literature categorizes each risk measure in a class depending on its axioms. The first of such classes to appear was the class of coherent risk measures. Soon after, in critique to positive homogeneity, \cite{follmer02}, and \cite{frittelli02} proposed the class of convex risk measures. New classes are constantly emerging in increasing diversity and generality, such as the class of law invariant, comonotone, and comonotonic convex risk measures. See \cite{kusuoka01,frittelli05,jouini06} and \cite{song09}. One of the broader classes is the monetary risk measures. See \cite{follmer16} for a review. 

Recently, \cite{castagnoli21} propose the class of star-shaped risk measures, which can be understood as a generalization for positive homogeneity and convexity as well as the star-shaped property of the generated acceptance set. This class is a significant generalization because it allows for the most used non-convex risk measure, the Value at Risk (VaR), to be in the same class as the convex risk measures. The authors found that a normalized monetary risk measure is star-shaped (see \Cref{defins}) if and only if it can be represented as the minimum of normalized convex risk measures. \cite{Liebrich2021} explores the role of star-shaped risk measures for inf-convolution problems of risk sharing.

Nonetheless, \cite{jia21} show that any monetary risk measure is the minimum of convex risk  measures, without any mention of star-shapedness. In fact, they found that a functional is a monetary risk measure if and only if it is the minimum of all comonotone convex risk measures that dominate it. Thus, there seems to be a missing link between the two approaches. Since both take a minimum of convex risk measures, but one obtains a monetary risk measure and the other a star-shaped risk measure, a direct question would be if any monetary risk measure is also star-shaped, perhaps under normalization.

In this letter, we show that the answer to this question is no, even the missing link being the acceptability of $0$ for the risk measure, which is implied by  normalization. In this sense, we make clear the link between the value for risk measures at $0$ and star-shapedness. Furthermore, we show that monetary risk measures are only at a translation away from star-shapedness.  Formally, we have the following claim:
\\
\\
\textit{Let $\p$ be a monetary risk measure. Then  $\p-c$ is star-shaped for some $c\in\R$  if and only if \[\p(X)=\min_{ \Lambda}(X)= \p_\lambda(X),\:\forall\:X,\] where $\Lambda$ is a family of convex risk measures  such that $\sup_{ \Lambda} \p_\lambda (0) < \infty$. In this case, we can take any $c\geq\sup_{ \Lambda} \p_\lambda (0)$.
}
\\
\\
Thus, the link between star-shaped and monetary risk measures  is the representation of the translated risk measure as pointwise minimum of some family of convex risk measures, under the mild condition that the pointwise supremum of such family at $0$ is finite.

In all that follows, $(\Omega,\mathfrak{F}, \Pro)$ is a fixed underlying probability space.
Every equality and inequality is to be understood as holding $\Pro$-almost surely. 
We work on $L^\infty := L^\infty(\Omega,\mathfrak{F},\Pro)$, ``the set of all ($\Pro$-equivalence classes of) random variables $X$ which are $\Pro$-essentially bounded''. $\E$ is the expectation under $\Pro$, and $1_{A}$ is the indicator function of event $A$. The remainder of this letter is structured as follows: section \ref{back} exposes a brief background on risk measures, with special attention to star-shapedness; section \ref{results} presents our main results, which are the role played by the acceptability of $0$ and the link between monetary and star-shaped risk measures.

\section{Background}\label{back}

We are interested in functionals that measure risk. In this sense, we expose below a formal definition of risk measures and the properties we shall consider in this letter.

\begin{defi}\label{defins}
A functional $\p: \X \rightarrow \R\cup \{-\infty,\infty\}$  may fulfill the following:
 
\begin{enumerate}[label = (\roman*)]

\item ({Monotonicity}) $\p$ is {monotone} if $ Y \leq X$ implies $\p(Y) \geq \p(X)$ for every pair $X,Y\in\X$ . 
\item ({Translation invariance}) $\p$ is {translation invariant} if $\p (X + c) = \p (X) - c$ for any $X\in\X$ and $c \in \R$.

\item (Normalization) $\p$ is normalized if $\p(0) =0$.
\item ({Convexity}) $\p$ is {convex} if $\p(\lambda X + (1- \lambda)Y) \leq \lambda \p(X) + (1- \lambda ) \p(Y)$, for every pair $X,Y\in\X$ and all $\lambda \in [0,1]$.
\item ({Positive homogeneity}) $\p$ is {positive homogeneous} if $\p(\lambda X) = \lambda  \p(X)$ for all $X\in\X$ and $\lambda \geq 0$.

\item (Star-shapedness) $\p$ is star-shaped if $\p(\lambda X ) \geq \lambda \p(X)$ for all $X \in \X$ and $\lambda \geq 1 $.


\end{enumerate}
We say the function is risk measure if it is not identically $\infty$ and does not attain $-\infty$. A monetary risk measure  satisfies translation invariance and monotonicity. A convex risk measure is a convex monetary risk measure, and a coherent risk measure is a positive homogeneous convex  risk measure. We say a risk measure is star-shaped  if it fulfills star-shapedness and normalized if it possesses normalization. 
\end{defi}

Star-shapedness is a generalization of positive homogeneity, which in its turn implies normalization. It is equivalent to  $\p(\lambda X ) \leq \lambda \p(X)$ for $0 < \lambda < 1 $. Thus, it is also implied by convexity if we have in addition $\p(0) \leq 0$, which is the case for normalized risk measures. For the intuition on most properties above, we recommend \cite{follmer16} and references therein. The financial meaning of star-shapedness is well discussed in \cite{castagnoli21}.

Any risk measure has a set of acceptable positions, the so-called acceptance set. In the same vein, one can take the acceptance set as the primary object and define a risk measure through it. 
Once the acceptance set is taken as the primary object, it is natural to care about its properties. We now define such concepts.

\begin{defi}
Given a risk measure $\p$,  its acceptance set is defined as \[\Ap:= \{ X \in \X : \p(X) \leq 0  \}.\] Additionally, given an acceptance set $\A \subseteq \X$, its induced risk measure $\pa$ is given by \[ \pa(X) = \inf \{ m \in \R : X + m \in \A\},\:\forall\:X\in L^\infty.\]

An acceptance set $\A \subseteq \X$ may have the following properties:
\begin{enumerate}[label=(\roman*)]
\item (Monotonicity) $\A$ is monotone if $X \in \A$ and $X\leq Y$, $Y \in \X$ implies $Y \in  \A$.

\item (Monetarity) $\A$ is monetary if is monotone and  $\inf \{ m \in \R: m \in \A \}  > -\infty$.

\item (Normalization) $\A$ is normalized if is monotone and  $\min \{ m \in \R: m \in \A \} = 0$.

\item (Convexity)  $ \A$ is convex if $X$ and $Y \in  \A$ implies $\lambda X +(1-\lambda) Y \in  \A$. 

\item (Conicity) $ \A$ is a cone if $X \in  \A$ implies $\lambda X \in  \A$ for all $\lambda \geq 0 $.

\item (Star-shapedness)  $ \A$ is a star-shaped at $\mathcal{V}$ if  $X \in \A$ implies  for all $v \in \mathcal{V}$  that $\lambda X + (1-\lambda)v \in  \A$ for all $\lambda \in (0,1)$. Unless otherwise specified, star-shapedness is to be understood as star-shapedness at $0$.

\end{enumerate}
\end{defi}

  We are actually introducing the nomenclature of normalized and monetary acceptance set. The intuition behind those properties can be found in \cite{follmer16}. Normalization was used in \cite{jia21}, but not under this appellation. It is worth noting that while star-shapedness is a relaxation of conicity, it is also implied by convexity under $0 \in \cl\A$, where $\cl$ means closure with respect to the essential supremum norm.

\begin{remark}
We use a broader definition of star-shapedness, as used in \cite{penot05} for example, because we use it as a tool. In the risk measures literature, only star-shapedness at $0$ is considered. The definition found in \cite{castagnoli21} is that $S$ is star-shaped if and only if $\lambda s \in S$ for all $s \in S$ and $\lambda \in [0,1]$.
 We opt to use the open interval $(0,1)$ for $\lambda$ instead of the closed $[0,1]$ since our sets do not necessarily contain $0$. Of course, $\lambda =1$ is always covered. Note that both definitions coincide for closed sets. Moreover, since we are on $\Li$, all monetary risk measures are Lipschitz continuous, which implies the acceptance sets $\Ap$ are closed. Our definition of star-shaped at $0$ implies that $0 \in \cl \A$ and both definitions are  equivalent if $0 \in \A$. 
\end{remark}

The following result shows the interplay between any risk measure $\p$ with its acceptance set $\A_\p$, as well as the interplay between any acceptance set $\A$ and its induced risk measure $\pa$.

\begin{propo}\label{basic props} (Propositions 2.1 and 2.2 of \citet{artzner99}, Propositions 4.6 and 4.7 of \citet{follmer16} and Proposition 2 of \cite{castagnoli21}). Given an acceptance set $\A \subseteq \X$ and a risk measure $\p : \X \rightarrow \R$, we have the following interplay.
\begin{enumerate}[label=(\roman*)]
\item (Monotonicity)  If $\p$ is monotone, then $\A_\p$ is also monotone. If $ \A$ is monotone, then $\pa $ is also monotone.

\item (Monetarity)
 If $\p$ is a monetary risk measure, then $\A_\p$ is monetary and  $\p = \p_{\A_\p}$. If $ \A$ is monetary, then $\pa$ is monetary and the closure of $ \A$ is equal to $\A_\pa $. 

\item (Normalization)
 If $\p$ is a normalized monetary risk measure, then $\A_\p$ is normalized.
If $ \A$ is normalized, then $\pa$ is a normalized monetary risk measure.

\item (Convexity) 
 If $\p$ is convex, then $\A_\p$ is convex. If $ \A$ is convex, then $\pa $ is convex. 

\item (Positive Homogeneity) If $\p$ is positive homogeneous, then $\A_\p$ is a cone. If  $ \A$ is a cone, then  $\pa $ is positive homogeneous. 

\item (Star-shapedness) 
 If $\p$ is star-shaped, then $\A_\p$ is star-shaped. If  $ \A$ is star-shaped, then  $\pa $ is star-shaped.

\end{enumerate}
\end{propo}

\begin{remark}
To take the closure of a monetary acceptance set does not alter the resulting risk measure. By the last Proposition we have $\cl \A = \A_\pa$ and $\pa = \p_{\A_\pa} = \p_{\cl A}$. Importantly, $\cl A = \A_\pa$ holds for the closure under the supremum norm, while this reasoning is not necessarily true for more general $\Lp$ spaces. For such spaces requiring the monetary risk measure to be lower semicontinuous in the natural $p$-norm is enough to maintain the results, see \cite{follmer16} for instance.	
\end{remark}

Regarding representations, \cite{jia21} show that any monetary risk measure is the pointwise minimum of a family of convex risk measures. In its turn, \cite{castagnoli21} show that any monetary normalized star-shaped risk measure is the pointwise minimum of normalized convex risk measures. Both theorems conclude that a positive homogeneous monetary risk measure is the pointwise minimum of a family of coherent risk measures. We now expose both results under our notation.

\begin{theorem}\label{theojia}(Theorems 3.1 and 3.4 of \cite{jia21}) A  functional  $\p : \X \rightarrow \R$ is a monetary risk measure if and only if \[\rho (X) = \min_{ \Lambda} \p_\lambda (X),\:\forall\:X\in\X, \]
	where $\Lambda$ is a set of convex monetary risk measures. Furthermore, $\p$ is positive homogeneous if and only if $\Lambda$  contains only coherent risk measures.
\end{theorem}

\begin{theorem}(Theorem 5 of \cite{castagnoli21})\label{teocast}  $\p : \X \rightarrow \R$ is a star-shaped normalized monetary risk measure if and only if \[\rho (X) = \min_{\Lambda} \p_\lambda (X),\:\forall\:X\in\X, \]
	where $\Lambda$ is a set of  normalized convex risk measures. Furthermore, $\p$ is positive homogeneous if and only if $\Lambda$  contains only coherent risk measures.
\end{theorem}

\begin{remark}
	Both studies provide many representations closely related to the dual conjugate of convex risk measures or as the lower envelop of a family of risk measures. In fact, if $\p$ is a monetary (normalized star-shaped) risk measure, one can take $\Lambda= \{ h(X) : h \text{ is a (normalized) convex risk measure and } h \geq \p \}$. Due to parsimony, we focus only on the main representations we expose since our results are also valid for all the remaining ones. Moreover, by the direct interplay of \Cref{basic props}, we take the liberty to move between $\p_\lambda$ and $\A_\lambda$ in $\Lambda$ without further mention.
\end{remark}

\section{Main Results}\label{results}

In this section, we will compare the main results of \cite{jia21} and \cite{castagnoli21}, with particular emphasis on the acceptance of $0$. The two theorems are strikingly similar. The only difference is that the minimum is star-shaped under normalization. The role of normalization is essential to obtain star-shapedness because the minimum of a family of convex risk measures is not necessarily star-shaped. Furthermore, under mild conditions, a monetary risk measure is only a translation away from star-shapedness. The next example illustrates this fact.

\begin{ex}
Let $\p_\lambda (X) := \p (X) + f(\lambda),\:\lambda\in\mathbb{R}$, where $\p:\X \rightarrow \R $ is  a convex risk measure and $f: \R \rightarrow \R_+$ satisfies $\min_{\lambda \in \R} f(\lambda) = f(\epsilon) > -\p(0) $ and $\lim_{\lambda \rightarrow \infty} f(\lambda) = \infty$. Note that $\p_\lambda$ is not normalized, while it is clearly a convex risk measure. We claim that $\min_{\lambda \in \R} \p_\lambda $ is not star-shaped. First of all, note that $\min_{\lambda \in \R} \p_\lambda (X) = \p(X) + \min_{\lambda \in \R} f(\lambda)  = \p_\epsilon (X) $. Moreover, for any $k >1$ and constant $X$, we have
 \[\p_\epsilon (k X)- k\p_\epsilon ( X) =  \p(0) -kX -k (\p(0) -X) + f( \epsilon ) (1-k) = \left( f(\epsilon) +\p(0)\right) (1-k) <0. \]
Hence, there is $X\in\X$ such that $\p_\epsilon (k X) <  k \p_\epsilon (X)$ for $k>1$, which implies that $\p_\epsilon$ is not a star-shaped risk measure. Moreover, if $\rho$ has in addition positive homogeneity, then the claim is valid for any $X\in L^\infty$ since $f(\epsilon)>-\p(0)=0$ implies \[\p_\epsilon (k X)- k\p_\epsilon ( X) =  \p(kX) -k\p(X) +(1-k) f( \epsilon )  = f(\epsilon) (1-k) <0. \] However, note that while this risk measure is not star-shaped it is only a translation away from it. In fact, the functional \[\p^* (X) := \min_{\lambda \in \R}\p_\lambda(X) - \min_{\lambda \in \R}\p_\lambda(0) =\p(X)-\p(0)\] is normalized and convex, which implies it is star-shaped. For a special case let $\p(X) = \E[-X]$ and $f(\lambda) := \lambda 1_{\lambda \geq \epsilon} + \epsilon 1_{\lambda < \epsilon}$  for some $\epsilon >0$. Then, for $\lambda \geq \epsilon$ we obtain $\p_\lambda (X) =  E[-X] + \lambda$ and \[\min_{\lambda \in \R} \p_\lambda (X)= \min_{\lambda \geq \epsilon} \p_\lambda (X) =\p_\epsilon(X),\] which is not star-shaped. 
\end{ex}

The role played by normalization becomes more evident when we look from the  perspective of acceptance sets. In fact, it is a sufficient condition for star-shapedness.

\begin{propo}
Consider a family $\Lambda$ of monetary convex normalized acceptance sets and $\A := \cup_{ \Lambda} \A_\lambda$. Then $\p_A$ is monetary normalized star-shaped risk measure. 
\end{propo}

\begin{proof}
 We have by \Cref{basic props}  that $\pa$ is star-shaped (convex) if and only if $\A$ is star-shaped (convex). Furthermore, for a family $\Lambda$ of normalized acceptance sets it holds that $\A := \cup_{\Lambda} \A_\lambda$ is also normalized and $\pa (X) = \inf_{ \Lambda} \p_{\A_\lambda} (X)$. We have that star-shapedness is preserved under union and convexity implies star-shapedness as long as the set contains $0$ in its closure, which is implied by normalization. Thus, if each $\A_\lambda$ is normalized and convex, then $ \A$ is star-shaped. Hence, by \Cref{basic props} we get that $\pa$ is also star-shaped.
\end{proof}

\begin{remark}
While normalization gives us sufficient conditions to generate star-shaped risk measures, it is not necessary. In fact, instead of demanding for each $\A_\lambda$ to be normalized we could simply ask them to contain $0$, which is equivalent to $ \p_{ \A_{ \lambda}}  (0) = \inf  \{ m \in \R : m \in \A_\lambda \} \leq 0$. This, together with convexity is enough to guarantee that $\pa$ is star-shaped.  Star-shapedness in its turn implies $\p(0)  \leq 0$ since for any $\lambda \in (0,1)$ we have $\p(0) = \p( \lambda 0) \leq  \lambda \p(0)$. Moreover, even if convexity is relaxed to star-shapedness, the result of this Proposition would hold. Hence, the previous Proposition indicates that while \cite{castagnoli21} demanded in their work that all risk measures be normalized, their result would remain true if it was replaced by the weaker assumption $\p(0) \leq 0$. 
\end{remark}

 We now provide a necessary condition for star shapedness of monetary risk measures.

\begin{propo}
If a monetary risk measure $\rho$ is star-shaped, then there exists a family $\Lambda$ of convex risk measures with at least one member  star-shaped such that \[\rho (X) = \min_{\Lambda} \p_\lambda (X),\:\forall\:X\in\X.\]	Moreover, $\Lambda$ can be taken as a family of comonotone convex risk measures.
\end{propo}

\begin{proof}
By \Cref{theojia}, as $\p$ is monetary there is a representation $\p (X) = \min_{ \Lambda} \p_\lambda (X),\:\forall\:X\in\X,$ under convex risk measures $\Lambda$. We only need to show that $\Lambda$ contains a star-shaped risk measure. This is equivalent to show  for one $\rho_\lambda \in \Lambda$  that $\rho_\lambda (0) \leq 0$.  As $\p$ is star-shaped we have that $0 \geq \p(0) =  \min_{ \Lambda} \p_\lambda (0) =  \p_{\lambda^*} (0)  $ for some $ \p_{\lambda^*} \in \Lambda$. Hence, $\p_{\lambda^*}$ is star-shaped. The last claim is a direct consequence of Theorem 6.1 in \cite{jia21}. 
\end{proof}

The subsequent Propositions show that any monetary risk measure is just at  a translation away from star-shapedness under the correct choice of acceptance sets.

\begin{propo}\label{prop1}
Let $\Lambda$ be a family of convex monetary acceptance sets, $\A:= \cup_{\Lambda} \A_\lambda$ and $\B:= \cl( \cap_{ \Lambda} \A_\lambda)$. Then  for any $Y \in B$, $\p_Y (X) := \pa (X +Y)$ is a star-shaped monetary risk measure. In particular, $\p(X) := \pa(X) - c $ is a star-shaped monetary risk measure for any $c\geq \p_\B (0)$.
\end{propo}

\begin{proof}
Firstly, note that $\B$ is a convex set, and $X \in \A$ or $X\in\B$ implies that $ X$ is contained in $\cl (\A_\lambda)$ for some $\lambda$. Thus, convexity of $\cl (\A_\lambda) $ yields that $k\A+(1-k)\B\subset\A$ for any $k \in (0,1)$, which is star-shapeness of $\A$ at $\B$. Note that if $\B=\emptyset$, then the claim is vacuously true. Furthermore, the set $ \A_Y = \A - Y = \{ X-Y : X \in \A\}$, with $Y \in  \B$, is star-shaped at $0$. In order to verify this claim note that for any $X \in \A_Y$ and $Y \in \B$ it follows that $ X + Y \in \A$. Star-shapedness of $\A$ at $Y$ implies that, for any $k \in (0,1)$, $k (X+Y) + (1-k)Y =kX+Y\in \A$. This is equivalent to  $kX \in \A_Y$, which is star-shapedness at $0$. In addition, we have that
\begin{align*}
 \p_{\A_Y} (X) &= \inf \{m \in \R: X+m \in \A_Y  \} 
 \\ &=  \inf \{m \in \R: X +m\in \A-Y  \}
 \\ & =\inf \{m \in \R: X +m+Y\in \A  \} 
 \\ &= \pa(X+Y) =\p_Y (X).
\end{align*} Thus,  by \Cref{basic props} we conclude that $\p_Y$ is star-shaped. To see that $\p$ is star-shaped, first note that $\p_\B(0) = \inf\{ m \in \R :m \in \B\} = \inf \{ m \in \R \cap \B\}= \min \{ m \in \R \cap \B\}$. The minimum is attained because $\B \cap \R$ is closed and bounded below. Thus, by Monotonicity, if $c\in\R$ and $c\geq \p_\B (0)$, $c\in\B$. 
Hence, by making $Y=c$ and considering the previous reasoning for $\p_Y$, we have that $\p_{c}$ is star-shaped. Since as all those risk measures  are monetary, we get for any $X\in\X$ that $ \p_{c}(X)=\pa(X+c)=\p(X)$. 
\end{proof}

\begin{propo}\label{rema}
	Let $\Lambda$ be a family of monetary acceptance sets, with their respective monetary risk measures, and $\B:= \cap_{ \Lambda} \A_\lambda$. Then $\B\cap\R\not=\emptyset$ if and only if $\B \neq \emptyset$ if and only if $\sup_{ \Lambda} \p_\lambda (0) =\p_\B(0)< \infty$. 
\end{propo}

\begin{proof}
 Since we are on $\X$ and the sets in $\Lambda$ are monotone,  $X \in  \B$ implies $\esssup X \in \B$. In other words,  $\B$ contains a constant, which is equivalent to $\B\cap\R\not=\emptyset$ if and only if $\B\not=\emptyset$. This fact leads to
\begin{align*}
	\sup_{ \Lambda} \p_\lambda (0) 
	= \inf \left\lbrace m \in \R : m \in \B\right\rbrace \leq \esssup X < \infty.
\end{align*} For the converse, if $\sup_{ \Lambda} \p_\lambda (0) < \infty$, then we get that \begin{align*}
	\infty>\sup_{ \Lambda} \p_\lambda (0) 
	= \inf \left\lbrace m \in \R : m \in \B\right\rbrace .
\end{align*} By monotonicity of $\B$, we have that $m\in\B\cap\R$ for any $m>\p_\B(0)$. Hence, $\B\not=\emptyset$.
\end{proof}

\begin{propo}\label{propo from star}
Let $\p$ be monetary. If $\A_\p$ is star-shaped at $Y\in\X$, then  $\A_\p = \A := \cup_{ \Lambda} \A_\lambda$ for a family $\Lambda$ of convex acceptance sets with $\B:= \cap_{ \Lambda} \A_\lambda\not=\emptyset$. Furthermore, $\p = \min_{ \Lambda } \p_{\A_\lambda}$, where each $\p_{\A_\lambda}$ is a convex risk measure.
\end{propo}

\begin{proof}
The fact that $\A_\p = \A$ is direct from \Cref{teocast}.  We need to show that there exists a family of such sets with a non-empty intersection. For each $Z \in \X$, let $\A_Z := \{ X \in \X : X \geq Z\}$. Note that each of such sets is  convex and monetary. Additionally, $\A_\p = \bigcup_{Z \in \A_\p} \A_Z$.  As $\A_\p$ is star-shaped at $Y$ and monotone, it is also star-shaped at $\A_Y$. Thus, since $\A_\p$ is closed we have $\A_Y\subset\A_\p$. Furthermore,  any member of $\az $ is contained in $\Ap$. Thus, we have that $ \text{conv}(\A_Z \cup \A_Y) \subset \A_\p$ for all $Z \in \Ap$. Now, let $\Lambda := \{  \text{conv}(\A_Z \cup \A_Y) : Z \in \Ap \}$. Clearly, $\A_\p = \A:= \cup_{ \Lambda} \A_{\lambda }$ and $Y \in \B := \cap_{ \Lambda} \A_{ \lambda} \neq \emptyset$. Thus, the fact that $\p = \inf_{ \Lambda } \p_{\A_\lambda}$ is directly obtained. Furthermore, it is straightforward from \Cref{basic props} that each $ \p_{\A_\lambda}$ is a convex risk measure. To show that the minimum is attained, firstly note that, for any $X\in\X$, $X + \p(X) \in \Ap$  and  $\A_{X+\p(X)} \subseteq    \text{conv} ( \A_{ X+\p(X)} \cup \A_Y) =: \A^* \in \Lambda $. Then, we have that 
\begin{align*}
 \p_{ \A^*}  (X) &= \inf \{ m \in \R: X+ m \in  \A^* \}  
 \\ &\leq \inf \left\lbrace  m \in \R: X+ m \in {\A_{X+\p(X)}} \right\rbrace   
 \\ &= \inf \{ m \in \R: X+ m \geq X + \p(X) \}  
 \\ &= \p(X).
\end{align*}
On the other hand, we have  $\p_{ \A^*}(X)\geq\inf_{\Lambda } \p_{\A_\lambda} (X)\geq\p(X)$. Hence, $\p_{ \A^*}(X)=\p(X)$, which implies that the minimum is attained. 
\end{proof}

We are now in condition to prove our main result, which is the following Theorem below. Such a claim unifies the discussion above by exposing the interplay between monetary and star-shaped risk measures.

\begin{theorem}\label{main theo}

Let $\p$ be a monetary risk measure. Then  $\p-c$ is star-shaped for some $c\in\R$  if and only if \[\p(X)=\min_{ \Lambda} \p_\lambda (X) ,\:\forall\:X\in\X,\] where $\Lambda$ is a family of convex risk measures such that $\sup_{ \Lambda} \p_\lambda (0) < \infty$. In this case, we can take any $c\geq\sup_{ \Lambda} \p_\lambda (0)$.

\end{theorem}

\begin{proof}
 For the if part, let $c \in\left( \cap_{\Lambda}\A_\lambda\right) \cap\R$. Note that such $c$ exists by \Cref{rema} since $\p_{ \cap_{\Lambda} \A_\lambda} (0)<\infty$. Thus, \Cref{prop1} implies that $\p-c$ is star-shaped.
For the only if part, as $\p-c$ is star-shaped, we have by \Cref{basic props} that
$\A_{\p-c} = \A_\p - c$ is also star-shaped and $\A_\p $ is star-shaped at $c$. Then, the representation as a minimum over $\Lambda$ is given by \Cref{propo from star}. This implies $\A_{\p-c}=\cup_{\Lambda}\A_\lambda$.
  By \Cref{rema},  the intersection of the acceptance sets in $\Lambda$ contains some constant, which implies in $\sup_{\Lambda} \p_\lambda (0) =  \p_{ \cap_{\Lambda} \A_\lambda} (0) < \infty$. Finally, by \Cref{prop1} we can take any $c \geq\p_{\cl (\cap_{\Lambda}\A_\lambda) }(0)=\p_{\cap_{\Lambda}\A_\lambda}(0)=\sup_{ \Lambda} \p_\lambda (0)$.
\end{proof}

We end this letter with an example that illustrates that not every monetary risk measure is at a translation away from star-shapedness, reinforcing the importance for the condition in \Cref{main theo}.

\begin{ex}
Let $f : \R \rightarrow \mathbb{Z}$ be the floor function defined as  $f(c) = \sup\{ m \in \mathbb{Z} : m \leq c \} $, and $\p^*$ a convex risk measure. Then $\p (X) := \p^* (f(X))$ is not star-shaped  under any translation. To verify that, note that for any constant $X$ we have $\p(X) =  \p^* (0)-f(X)$. Then, to show that it is not star-shaped under any translation is equivalent to show that for all $k \in \R$, which is the candidate for translation, there is some $\lambda \in (0,1)$ and some constant $X$ such that
\[ \dfrac{f(\lambda X) - \lambda f(X) }{1-\lambda} < k+\p^*(0). \]
By choosing $X=1$  and taking $\lambda \rightarrow 1^-$ yields,
\[ \lim_{\lambda \rightarrow 1^-} \dfrac{f(\lambda ) - \lambda f(1) }{1-\lambda} = \lim_{\lambda \rightarrow 1^-} \dfrac{ - \lambda  }{1-\lambda}  = -\infty.  \]
Hence, for any $k \in \R$ we can find a $\lambda \in (0,1)$ such that $ \dfrac{f(\lambda 1) - \lambda f(1) }{1-\lambda} < k+\p^*(0)$. Importantly, one should note that $\p$ is not translation invariant. However, $\p_\Ap$ is, and $\Ap = \A_{\p_\Ap}$. Therefore, $\p_\Ap$ is monetary, but not star-shaped risk measure.
 From this,  \Cref{main theo} implies that for any family $\Lambda$ of convex sets such that $\A_\p = \cup_{\Lambda} \A_\lambda$, it follows that $ \cap_{  \Lambda} \A_\lambda = \emptyset$. In order to complement this reasoning, in \Cref{fig} we expose the acceptance set of some coherent risk measure $\p^*$  in red, and  the acceptance set of $\p (X) := \p^* (f(X))$ in blue. It is visually clear that it does not matter how much we shift $\A_\p$, which translates $\p$, it  will never be star-shaped.

\begin{figure}[H]
\caption{In red is the  acceptance set of some coherent risk measure and in blue is the acceptance set of the coherent risk measure composed with the floor function.
}\label{fig}
\begin{center}
\begin{tikzpicture}[scale=1]

\clip (0,0) (-4,-2)rectangle (4,4);

\draw[thin,<->] (-4,0) -- (4,0);
\draw[thin,<->] (0,-2) -- (0,4);

\draw[ultra thick,smooth,color=red!70!black,domain=-4:0] plot (\x, { -(\x)}) ;
\draw[ultra thick,smooth,color=red!70!black,domain=0:4] plot (\x, { -(0.5*\x)}) ;

\draw[ultra thick,color=Blue!70!black] (-4,3.98)--(-3,3.98)--(-3,3)--(-2,3)--(-2,2)--(-1,2)--(-1,1)--(0,1)--(0,0)--(2,0)--(2,-1)--(3.98,-1)--(3.98,-2);

\fill[color=red, opacity=.2, smooth](-4,4)--(0,0)--(4,-2)--(4,4) ;

\fill[color=Blue, opacity=.35, smooth] (-4,4)--(-3,4)--(-3,3)--(-2,3)--(-2,2)--(-1,2)--(-1,1)--(0,1)--(0,0)--(2,0)--(2,-1)--(4,-1)--(4,-2)--(4,4);

\fill (0,0) circle [radius=0.03125cm] node[anchor=north west,scale=.8]{$0$};

\end{tikzpicture}
\end{center}
\end{figure}
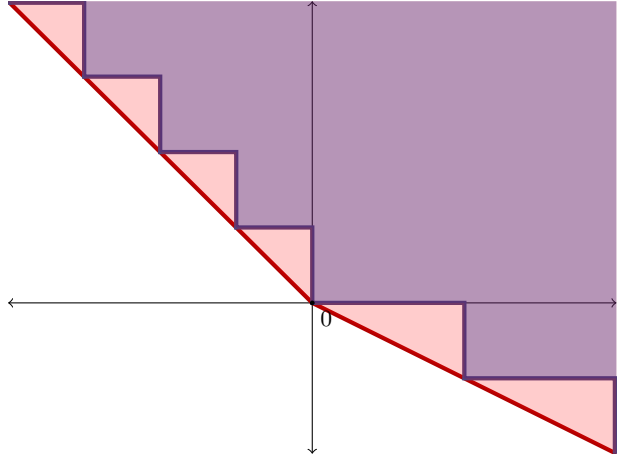

\end{ex}

\bibliography{ref}
\bibliographystyle{apa}
\end{document}